\documentclass[12pt]{article}
\usepackage{amsmath, amsthm, amssymb}
\usepackage{setspace}
\usepackage[usenames]{color}

\newtheorem{theorem}{Theorem}
\newtheorem{definition}{Definition}
\newtheorem{proposition}{Proposition}
\newtheorem{lemma}{Lemma}

\newcommand{\<}{\prec}
\newcommand{\cle}{\preceq}
\renewcommand{\>}{\succ}

\newcommand{\al}{\alpha}

\newcommand{\la}{\lambda}

\newcommand{\NN}{\mathbb N}

\newcommand{\RR}{\mathbb R}

\newcommand{\mc}[1]{\mathcal #1}

\newcommand{\MMMM}{\mc M}
\newcommand{\PPPP}{\mc P}

\newcommand{\Ra}{\ \Rightarrow \ }

\newcommand{\lemref}[1]{Lemma~\ref{lem:#1}}

\newcommand{\propref}[1]{Proposition~\ref{prop:#1}}

\newcommand{\eqrefs}[2]{Equations~\eqref{eq:#1} and~\eqref{eq:#2}}

\newcommand{\menge}[2]{\left\{{#1}\ \big|\ {#2}\right\}}
\newcommand{\ie}{i.\,e.\ }
\newcommand{\eg}{e.\,g.\ }

\DeclareMathOperator{\supp}{supp}

\begin{document}

\title{Oracle Tractability of Skew Bisubmodular Functions}

\author{Anna Huber and Andrei Krokhin\\
Durham University, UK}
\date{}
\maketitle 

\begin{abstract}
In this paper we consider skew bisubmodular functions as introduced in
\cite{ours:soda}.
We construct a convex extension of a skew bisubmodular function which we call Lov\'asz extension in correspondence to the submodular case. We use this extension to show that skew bisubmodular functions given by an oracle can be minimised in polynomial time.
\end{abstract}

\section{Introduction}
A key task in combinatorial optimisation is the minimisation of discrete functions. Important examples are {\em submodular} functions, see \eg \cite{Fujishige:book,Lovasz,McCormick,Schrijver},
and {\em bisubmodular} functions, see \eg  \cite{BouchetCunningham,Fujishige:book,McCormick,Qi}. 
These functions can be viewed as (special) functions from $D^n$ to $\RR$ where $D$ is a 2-element set for the submodular case and a 3-element set for the bisubmodular case. Fix a finite set $D$. One says that a class $\mathcal{C}$ of functions from $D^n$ to $\mathbb{Q}$ is {\em oracle-tractable} if there is an algorithm which, given a function $f\in C$ represented by a value-giving oracle, finds the minimiser of $f$ in time polynomial time in $n$
(the arity of $f$).
The oracle tractability of submodular and bisubmodular functions has been shown in \cite{Groetschel,Lovasz} and \cite{Qi} respectively, with many subsequent improvements (see \eg \cite{McCormick}).
Results about oracle tractability for other classes of discrete functions can be found in \cite{KrokhinLarose,Kuivinen}.

Submodular and bisubmodular functions play an important role for classifying the complexity
of optimisation problems known as {\em valued constraint satisfaction problems} (VCSPs).
These problems amount to minimising certain discrete functions represented as sums of bounded-arity functions.
Submodularity characterises tractable VCSPs on a two-element domain \cite{Cohen06:soft}. In \cite{ours:soda} a generalisation of bisubmodularity, {\em skew  bisubmodularity}, is introduced and used to classify the complexity of VCSPs on a three-element domain.
The tractability of skew bisubmodular function minimisation in the VCSP setting (i.e. represented as sums of bounded-arity skew bisubmodular functions) follows from \cite{Thapper12:power}, but the question whether skew bisubmodular functions are also tractable in the oracle model has been left open in \cite{ours:soda}. In this paper we construct a convex extension of a skew bisubmodular function, called Lov\'asz extension in correspondence to the submodular case \cite{Lovasz}, and show the oracle tractability of skew bisubmodular functions.

Very closely related results have recently appeared in \cite{Fujishige13}, where the authors acknowledge this work.

\subsection{Notation and Definition}
Skew bisubmodularity, also known as {\em $\alpha$-bisubmodularity}, is defined for functions $f:D^n\rightarrow \RR$ where $|D|=3$. In
\cite{ours:soda}, the elements of $D$ are denoted by $-1,0,1$. In this paper, we will fix $\al \in (0,1]$ throughout and, for convenience of notation, denote the elements of $D$ by $-\al,0,1$, replacing the name $-1$ by $-\alpha$. Obviously, 
 there is a direct correspondence between functions over $\{-1,0,1\}$ and functions over $\{-\al,0,1\}$. The definition of $\alpha$-bisubmodularity as in \cite{ours:soda} is then as follows.
Let $n \in \NN$. We write $[n] := \{1, \dots, n\}$.

Define the order $\<$ on $D$ through $0\<1$,  $0\<-\alpha$ and $1$ and $-\alpha$ being incomparable.
We also denote the corresponding component-wise order on $D^n$ by $\<$.

Define the binary operation $\wedge_0$ on $D$ as follows.
 \[\begin{array}{l}
1\wedge_0-\al = -\al\wedge_0 1=0; \\
x\wedge_0 y=\min(x,y) \mbox{ with respect to the above order if } \{x,y\}\ne \{-\al,1\}.
   \end{array}
\]
For $a\in D$, define the binary operation $\vee_a$ as follows:
 \[\begin{array}{l}
     1\vee_a-\al = -\al\vee_a 1=a; \\
x\vee_a y=\max(x,y) \mbox{ with respect to the above order if } \{x,y\}\ne \{-\al,1\}.
   \end{array}
\]
We also denote the corresponding component-wise operations on $D^n$ by $\wedge_0$ and $\vee_a$ respectively.

\begin{definition}\label{def:alpha}
A function $f:D^n\rightarrow \RR$ is called {\em $\alpha$-bisubmodular
} if, for all $\mathbf{a},\mathbf{b} \in D^n$,
\begin{equation}\label{alpha}
f({\bf a}\wedge_0 {\bf
b})+\alpha\cdot f({\bf a}\vee_0 {\bf b}) + (1-\alpha)\cdot f({\bf a}\vee_1 {\bf b}) \leq f({\bf a})+f({\bf b}).
\end{equation}
\end{definition}

The above inequality defines submodular functions if we restrict $D$ to $\{0,1\}$ (i.e. ignore $-\al$) and bisubmodular functions if $\al=1$.

\subsection{Result}
\begin{theorem}\label{thm:main}
There exists an algorithm that finds a minimum of any $\alpha$-bisubmodular function $f:D^n\rightarrow \mathbb{Q}$ in time polynomial in $n$ if $f$ is given by an oracle.
\end{theorem}
\begin{proof}
In the remainder of the paper we will construct for any $\alpha$-bisubmodular function $f:D^n\rightarrow \mathbb{Q}$ a convex extension $f^L:[-\al,1]^n\rightarrow \RR$ which takes its minimal value on $D^n$ and which can be efficiently computed on every rational vector in $[-\al,1]^n$. The theorem then follows from convex optimisation techniques, in the same way that sub- and bisubmodular minimisation are achieved through convex optimisation, see \cite{Lovasz} and \cite{Qi} respectively.
\end{proof}

\section{Lov\'asz Extension for Skew Bisubmodular Functions}

For ${\bf x} \in [-\al,1]^n$ let $\PPPP({\bf x})$ be the set of all probability distributions on $D^n$ with marginals ${\bf x}$, \ie
 $$\PPPP({\bf x}) := \menge{\la: D^n\rightarrow [0,1]}{\sum_{{\bf a} \in D^n}\la({\bf a}) = 1, \sum_{{\bf a} \in D^n}\la({\bf a}) {\bf a} = {\bf x}}$$

\begin{definition}[Lov\'asz Extension]
For a function $f:D^n\rightarrow \RR$ define the {\em Lov\'asz Extension}
$f^L:[-\al,1]^n\rightarrow \RR$ through
$$f^L({\bf x}) := \sum_{{\bf a} \in D^n}\la_{{\bf x}}({\bf a})f({\bf a}),$$
where $\la_{{\bf x}}$ is the unique element of $\PPPP({\bf x})$ such that its support forms a chain in $D^n$ with respect to the order $\<$. 
(The existence of this element is proved below in \lemref{lambda}).
\end{definition}

Note that, for ${\bf a} \in D^n$, one has $\la_{{\bf a}}({\bf a}) = 1$ and thus $f^L({\bf a}) = f({\bf a})$, \ie $f^L$ is indeed an extension of $f$.
It also follows directly from the definition that
$$\min \menge{f({\bf a})}{{\bf a} \in D^n} = \min \menge{f^L({\bf x})}{{\bf x} \in [-\al,1]^n}.$$

The restriction of $f^L$ to $[0,1]^n$ is the ordinary Lov\'asz extension for $f|_{\{0,1\}^n}$, as in \cite{Lovasz}.
In the case $\al = 1$, the function $f^L$ is the Lov\'asz extension for bisubmodular functions as in \cite{Qi}.

\begin{lemma}\label{lem:lambda}
For every ${\bf x} \in [-\al,1]^n$, there is a unique element $\la_{\bf x}$ of $\PPPP({\bf x})$ such that its support forms a chain in $D^n$ with respect to the order $\<$.
\end{lemma}
\begin{proof}
Let ${\bf x} \in [-\al,1]^n$ and write ${\bf x} = (x_1, \dots, x_n)$.\\
{\bf Construction:}
We will construct an element $\la_{\bf x} \in \RR^{D^n}$ and show that it has the required properties. To this aim we will  recursively construct two sequences, $({\bf u_i})_{i \in \NN}$ in $D^n$ and $({\bf x_i})_{i \in \NN}$ in $[-\al,1]^n$. For every $i \in \NN$ we write ${\bf u_i} = (u_{i1}, \dots, u_{in})$ and ${\bf x_i} = (x_{i1}, \dots, x_{in})$.

Let ${\bf x_1} := {\bf x}.$
Assuming that ${\bf x_i}$ is already constructed for some $i \in \NN$, we will construct ${\bf u_i}$ and ${\bf x_{i+1}}$ as follows.

Denote
$N_i := \menge{j \in [n]}{x_{ij}<0}$,
$Z_i := \menge{j \in [n]}{x_{ij}=0}$, and
$P_i := \menge{j \in [n]}{x_{ij}>0}$.
Let
$$u_{ij} :=
      \begin{cases}
       -\al &\mbox{for  }\ j \in N_i\\
       0    &\mbox{for  }\ j \in Z_i\\
       1    &\mbox{for  }\ j \in P_i , \\
      \end{cases}$$

$$\la_{\bf x}({\bf u_i}) :=
      \begin{cases}
 \min \left\{ \min \menge{-\frac{x_{ij}}{\al}}{j \in N_i}  ,  \min \menge{x_{ij}}{j \in P_i} \right\} &\mbox{if  }\ \bf u_i \neq 0\\
        1 -\la_{\bf x}({\bf u_1}) - \dots -\la_{\bf x}({\bf u_{i-1}}) &\mbox{if  }\ \bf u_i = 0
      \end{cases}$$
and
\begin{equation}\label{eq:3}
{\bf x_{i+1}} := {\bf x_i} - \la_{\bf x}({\bf u_i}){\bf u_i}.
\end{equation}
From this construction we have for every $j \in [n]$ that
\begin{align*}
&u_{ij}=0  &&&\Ra& x_{i+1,j}=0&\Ra& u_{i+1,j}=0\\
&u_{ij}=1 &\Ra& \la_{\bf x}({\bf u_i}) \le x_{ij} &\Ra& x_{i+1,j}\ge 0 &\Ra& u_{i+1,j} \in \{0, 1\}\\
&u_{ij}=-\al &\Ra& \la_{\bf x}({\bf u_i}) \le -\tfrac{x_{ij}}{\al} &\Ra& x_{i+1,j}\le 0 &\Ra& u_{i+1,j} \in \{0, -\al\},
\end{align*}
so $u_{i+1,j} \cle u_{ij}$ and thus ${\bf u_{i+1}} \cle {\bf u_i}$. Furthermore, if $\bf u_i \neq 0$ and $m \in [n]$ is such that either
$$m \in N_i\ \ \mbox{and} \ \ -\tfrac{x_{im}}{\al} = \min \menge{-\tfrac{x_{ij}}{\al}}{j \in N_i}  =\la_{\bf x}({\bf u_i})$$
$$\mbox{or}\ \ \ m \in P_i\ \ \mbox{and} \ \ x_{im} = \min \menge{x_{ij}}{j \in P_i} =\la_{\bf x}({\bf u_i}),$$
then $x_{i+1,m}=0$ and thus $u_{i+1,m}=0$, whereas $u_{im} \neq 0$. Thus ${\bf u_{i+1}} \< {\bf u_i}$.

Clearly, this recursive construction yields ${\bf u_{n+1}}=0$. Let
$k \in \NN$ be such that $\bf u_{k-1} \neq 0$ and $\bf u_{k}=0$ and let
$\la_{\bf x}({\bf v}) := 0$ for all ${\bf v} \in D^n \setminus \{{\bf u_1}, \dots, {\bf u_k}\}$.
The construction yields that the support of $\la_{\bf x}$ forms a chain in $D^n$ with respect to the order $\<$.
We will now prove that $\la_{\bf x} \in \PPPP({\bf x})$.

The choice of $k$ yields $\la_{\bf x}({\bf u_1}), \dots,\la_{\bf x}({\bf u_{k-1}})  \neq 0$.
Equation \eqref{eq:3} yields
\begin{equation}\label{eq:7}
\sum_{i=1}^{k-1}\la_{\bf x}({\bf u_i}){\bf u_i} ={\bf x}.
\end{equation}
Let $j \in [n]$ be such that $u_{k-1,j} \neq 0$. As ${\bf 0} \< {\bf u_{k-1}} \< \dots \< {\bf u_1}$, one has $u_{k-1,j} = \dots = u_{1j}$
and thus
$$\sum_{i=1}^{k-1}\la_{\bf x}({\bf u_i})u_{ij} = x_j$$
from  \eqref{eq:7} yields
$$\sum_{i=1}^{k-1}\la_{\bf x}({\bf u_i}) = \frac{x_j}{u_{1j}} \le 1.$$
If
$$\sum_{i=1}^{k-1}\la_{\bf x}({\bf u_i}) = 1,$$
then $\la_{\bf x}({\bf u_k}) = 0$ by definition and $\la_{\bf x}$ is supported by the chain $\{{\bf u_1}, \dots, {\bf u_{k-1}}\}$.
If
$$\sum_{i=1}^{k-1}\la_{\bf x}({\bf u_i}) < 1,$$
then $\la_{\bf x}({\bf u_k}) > 0$ by definition and
$\la_{\bf x}$ is supported by the chain $\{{\bf u_1}, \dots, {\bf u_k}\}$. One has

$$\sum_{{\bf a} \in D^n}\la_{\bf x}({\bf a}) = \sum_{i=1}^{k}\la_{\bf x}({\bf u_i}) = 1$$
 by definition and
$$\sum_{{\bf a} \in D^n}\la_{\bf x}({\bf a}) {\bf a} = \sum_{i=1}^{k}\la_{\bf x}({\bf u_i}){\bf u_i} = \sum_{i=1}^{k-1}\la_{\bf x}({\bf u_i}){\bf u_i}\stackrel{\eqref{eq:7}}{=} {\bf x},$$
so $\la_{\bf x} \in \PPPP({\bf x}).$\\
{\bf Uniqueness:}
Let $({\bf u_i})_{i \in \NN}$, $({\bf x_i})_{i \in \NN}$ and $\la_{\bf x}$ be as constructed above, let ${\bf v_1}\> \dots \> {\bf v_{\ell}}$ be a chain in $D^n$ and let $\mu \in \PPPP({\bf x})$ have support $\{{\bf v_1}, \dots, {\bf v_{\ell}}\}$. We will show that $\mu = \la_{\bf x}$. One has
\begin{equation}\label{eq:2}
\sum_{i=1}^{\ell}\mu({\bf v_i}){\bf v_i} = {\bf x}.
\end{equation}

Let $j \in [n]$.
As ${\bf v_1}\> \dots \> {\bf v_{\ell}}$, unless $v_{1j}= 0$, there is a $h \in [\ell]$ such that $v_{1j}= \dots = v_{hj} \neq 0$ and either $h = \ell$ or $ v_{hj}  \> v_{h+1,j} = \dots = v_{\ell j} = 0$. If $v_{1j}= 0$, Equation \eqref{eq:2} yields $x_{1j} = 0$ and thus $u_{1j} = 0$ by definition of $u_{1j}$. Otherwise, we have
\begin{equation}\label{eq:9}
v_{1j}\sum_{i=1}^{h}\mu({\bf v_i}) = \sum_{i=1}^{h}\mu({\bf v_i})v_{ij} = \sum_{i=1}^{\ell}\mu({\bf v_i})v_{ij} \stackrel{\eqref{eq:2}}{=} x_j.
\end{equation}
As $\sum\limits_{i=1}^{h}\mu({\bf v_i}) > 0$, the numbers $v_{1j}$, $u_{1j}$ and $x_j$ all have the same sign. Since $v_{1j}, u_{1j} \in \{-\al,0,1\}$, it must hold that $v_{1j} = u_{1j}$.
This yields ${\bf v_1} = {\bf u_1}$.

If $\ell=1$, we are done, as $\mu$ and $\la_{\bf x}$ both take the value $1$ on ${\bf v_1} = {\bf u_1}$ and $0$ otherwise, so $\mu = \la_{\bf x}$. If $\ell > 1$, let $m \in [\ell -1]$ be such that ${\bf v_h} = {\bf u_h}$ holds for all $h \le m$ and $\mu({\bf v_h}) = \la_{\bf x}({\bf u_h})$ holds for all $h < m$. We will show that $\mu({\bf v_m}) = \la_{\bf x}({\bf u_m})$ and ${\bf v_{m+1}} = {\bf u_{m+1}}$.

As ${\bf v_m} \> {\bf v_{m+1}}$ there is a $j \in [n]$ such that $v_{m+1,j}=0$ but $v_{mj}\neq 0$.

As ${\bf v_1}\> \dots \> {\bf v_{\ell}}$, one has $v_{1j}= \dots = v_{mj} \> v_{m+1,j} = \dots = v_{\ell j} = 0$, and thus
\begin{eqnarray*}
 \mu({\bf v_m})v_{mj}
 &=& \sum_{i=1}^{m}\mu({\bf v_i})v_{ij} - \sum_{i=1}^{m-1}\mu({\bf v_i})v_{ij}\\
 &=& \sum_{i=1}^{\ell}\mu({\bf v_i})v_{ij} - \sum_{i=1}^{m-1}\la_{\bf x}({\bf u_i})u_{ij}\\
 &\stackrel{\eqref{eq:2}, \eqref{eq:3}}{=}&  x_j - (x_j - x_{mj})\\
 &=& x_{mj}
\end{eqnarray*}
So if $v_{mj} = 1$ we must have $\mu({\bf v_m}) = x_{mj}$ and if $v_{mj} = -\al$ we must have $\mu({\bf v_m}) = -\frac{x_{mj}}{\al}$.

If $\mu({\bf v_m}) \neq \min \left\{ \min \menge{-\frac{x_{mj}}{\al}}{j \in N_i}  ,  \min \menge{x_{mj}}{j \in P_i} \right\} = \la_{\bf x}({\bf u_m})$ we get a contradiction to \eqref{eq:2} as then $\mu({\bf v_m}) > \la_{\bf x}({\bf u_m})$, and so, for $j' \in [n]$ such that $u_{(m+1)j'}=0$ but $u_{mj'}\neq 0$ we get the following.
As ${\bf u_1}\> \dots \> {\bf u_{k}}$, one has $u_{1j'}= \dots = u_{mj'} \> u_{m+1,j'} = \dots = u_{kj'} = 0$.

If $u_{mj'} = 1$, then $v_{1j'}= \dots = v_{mj'} = u_{1j'}= \dots = u_{mj'} = 1$ and $v_{(m+1)j'}, \dots , v_{\ell j'} \in \{0, 1\}$, and so we have
\begin{align*}
\sum_{i=1}^{\ell}\mu({\bf v_i}) v_{ij'} &\geq \sum_{i=1}^{m}\mu({\bf v_i}) v_{ij'} = \sum_{i=1}^{m}\mu({\bf v_i})\\    &> \sum_{i=1}^{m} \la_{\bf x}({\bf u_i}) = \sum_{i=1}^{m} \la_{\bf x}({\bf u_i}) u_{ij'}= \sum_{i=1}^{k} \la_{\bf x}({\bf u_i}) u_{ij'} = x_{j'},
\end{align*}
contradiction to \eqref{eq:2}.

Equally, if $u_{mj'} = -\al$, we have $v_{1j'}= \dots = v_{mj'} = u_{1j'}= \dots = u_{mj'} = -\al$ and $v_{(m+1)j'}, \dots , v_{\ell j'} \in \{0, -\al\}$, and so
\begin{align*}
\sum_{i=1}^{\ell}\mu({\bf v_i}) v_{ij'} &\leq \sum_{i=1}^{m}\mu({\bf v_i}) v_{ij'} = -\al\sum_{i=1}^{m}\mu({\bf v_i})\\    &<  -\al\sum_{i=1}^{m} \la_{\bf x}({\bf u_i}) = \sum_{i=1}^{m} \la_{\bf x}({\bf u_i}) u_{ij'}= \sum_{i=1}^{k} \la_{\bf x}({\bf u_i}) u_{ij'} = x_{j'},
\end{align*}
contradiction to \eqref{eq:2}.
We thus have $\mu({\bf v_m}) = \la_{\bf x}({\bf u_m})$. The fact that ${\bf v_h} = {\bf u_h}$ and $\mu({\bf v_h}) = \la_{\bf x}({\bf u_h})$ holds for all $h \le m$ implies ${\bf v_{m+1}} = {\bf u_{m+1}}$ by a similar argument as used to show ${\bf v_1} = {\bf u_1}$ in \eqref{eq:9}. This finishes the inductive proof that ${\bf v_h} = {\bf u_h}$ for all $h \in [\ell]$ and that $\mu = \la_{\bf x}$.
\end{proof}

\subsection{Convex Closure}
As, for every ${\bf x} \in [-\al,1]^n$, the set $\PPPP({\bf x})$ is a compact and non-empty subset of $\RR^{D^n}$, the set
$$\menge{\sum_{{\bf a} \in D^n}\la({\bf a})f({\bf a})}{\la \in \PPPP({\bf x})}$$
is a compact and non-empty subset of $\RR$, and so contains its infimum.

\begin{definition}[Convex Closure]
For a function $f:D^n\rightarrow \RR$ we define the {\em convex closure}
$f^-:[-\al,1]^n\rightarrow \RR$ by
$$f^-({\bf x}) := \min \menge{\sum_{{\bf a} \in D^n}\la({\bf a})f({\bf a})}{\la \in \PPPP({\bf x})}.$$
\end{definition}
\begin{proposition}\label{prop:convex}
 $f^-$ is convex.
\end{proposition}
\begin{proof}
Let $\beta \in (0, 1)$ and ${\bf x}, {\bf y} \in [-\al,1]^n$.
Let $\mu \in \PPPP({\bf x})$ be such that
$$f^-({\bf x}) = \sum_{{\bf a} \in D^n}\mu({\bf a})f({\bf a})$$
and let $\nu \in \PPPP({\bf y})$ be such that
$$f^-({\bf y}) = \sum_{{\bf a} \in D^n}\nu({\bf a})f({\bf a}).$$
Then $\beta \mu + (1-\beta)\nu \in \PPPP(\beta{\bf x} + (1-\beta){\bf y})$, and so
\begin{eqnarray*}
f^-(\beta{\bf x} + (1-\beta){\bf y}) &=&  \min \menge{\sum_{{\bf a} \in D^n}\la({\bf a})f({\bf a})}{\la \in \PPPP(\beta{\bf x} + (1-\beta){\bf y})}\\
 &\leq& \sum_{{\bf a} \in D^n}(\beta \mu + (1-\beta)\nu)({\bf a})f({\bf a})\\
  &=& \beta \sum_{{\bf a} \in D^n} \mu ({\bf a})f({\bf a}) + (1-\beta) \sum_{{\bf a} \in D^n}\nu({\bf a})f({\bf a}) \\
  &=&  \beta f^-({\bf x}) + (1-\beta)f^-({\bf y}).
\end{eqnarray*}
\end{proof}

\subsection{Convexity of the Lov\'asz Extension}

The following lemma generalises the corresponding results for submodular and bisubmodular functions, see \cite{Lovasz} and \cite{Qi}.

\begin{lemma}\label{lem:main}
The Lov\'asz extension $f^L$ is convex if and only if $f$ is $\alpha$-bisubmodular.\end{lemma}
\begin{proof}
Let $\mathbf{a},\mathbf{b} \in D^n$. If $f^L$ is convex, it holds that
\begin{equation}\label{eq:5}
 f^L\left( \tfrac{\mathbf{a}+\mathbf{b}}{2} \right) \le  \tfrac{f^L(\mathbf{a}) + f^L(\mathbf{b})}{2} = \tfrac{f(\mathbf{a}) + f(\mathbf{b})}{2}.
\end{equation}
It is easy to check that
\begin{equation}\label{eq:8}
({\bf a}\wedge_0 {\bf b})+\alpha({\bf a}\vee_0 {\bf b}) + (1-\alpha)({\bf a}\vee_1 {\bf b}) = \mathbf{a} + \mathbf{b},
\end{equation}
and so the probability distribution $\la$ with $\la({\bf a}\wedge_0 {\bf b}) = \tfrac{1}{2}$, $\la({\bf a}\vee_0 {\bf b}) = \tfrac{\alpha}{2}$ and $\la({\bf a}\vee_1 {\bf b}) = \tfrac{(1-\alpha)}{2}$ is in $\PPPP( \tfrac{\mathbf{a} + \mathbf{b}}{2})$. Furthermore, we have
$${\bf a}\wedge_0 {\bf b} \cle {\bf a}\vee_0 {\bf b} \cle {\bf a}\vee_1 {\bf b},$$
which means that $\la = \la_{\frac{\mathbf{a}+\mathbf{b}}{2}}$ and thus the value of the the Lov\'asz Extension at $\frac{\mathbf{a}+\mathbf{b}}{2}$ is
\begin{equation}\label{eq:6}
f^L\left( \tfrac{\mathbf{a}+\mathbf{b}}{2} \right) =  \tfrac{1}{2}f({\bf a}\wedge_0 {\bf b})+\tfrac{\alpha}{2}f({\bf a}\vee_0 {\bf b}) + \tfrac{(1-\alpha)}{2}f({\bf a}\vee_1 {\bf b}).
\end{equation}
\eqrefs{5}{6} imply \eqref{alpha}, so $f$ is $\alpha$-bisubmodular.

On the other hand, let $f$ be $\alpha$-bisubmodular. We will show $f^L = f^-$, as then $f^L$ is convex by \propref{convex}.

Let ${\bf x} \in [-\al,1]^n$. We will show $f^L({\bf x}) = f^-({\bf x})$.

Let
$$\MMMM({\bf x}) := \menge{\la \in \PPPP({\bf x})}{\sum_{{\bf a} \in D^n}\la({\bf a})f({\bf a}) = f^-({\bf x})}.$$
For every ${\bf a} = (a_1, \dots, a_n) \in D^n$ denote $z({\bf a}) := \left|\menge{i \in [n]}{a_i = 0}\right|.$
As $\MMMM({\bf x})$ is a compact and non-empty subset of $\RR^{D^n}$, the set
$$\menge{\sum_{{\bf a} \in D^n}\la({\bf a})z^2({\bf a})}{\la \in \MMMM({\bf x})}$$
is a compact and non-empty subset of $\RR$ and contains its supremum. Let $\mu \in \MMMM({\bf x})$ be such that
$$\sum_{{\bf a} \in D^n}\mu({\bf a})z^2({\bf a}) = \max\menge{\sum_{{\bf a} \in D^n}\la({\bf a})z^2({\bf a})}{\la \in \MMMM({\bf x})}.$$

To show $f^L({\bf x}) = f^-({\bf x})$, it is left to show that $\mu = \la_{\bf x}$. By \lemref{lambda} it suffices to show that $\mu$ is supported by a chain.

Assume that $\supp(\mu)$ is not a chain, and let ${\bf a}, {\bf b} \in \supp(\mu)$ be incomparable. We will define a function $\nu$ to contradict the choice of $\mu$. As $f$ is $\alpha$-bisubmodular, we have
\begin{equation}\label{eq:1}
 f({\bf a}\wedge_0 {\bf b})+\alpha\cdot f({\bf a}\vee_0 {\bf b}) + (1-\alpha)\cdot f({\bf a}\vee_1 {\bf b}) \leq f({\bf a})+f({\bf b}).
\end{equation}
Let $r:= \min\left\{\mu({\bf a}),\ \mu({\bf b}),\ \frac{1 - \mu({{\bf a}\wedge_0 {\bf b}})}{1+\al},\ 1 - \mu({\bf a}\vee_0 {\bf b}),\ 1 - \mu({\bf a}\vee_1 {\bf b})\right\}$.
Then $r>0$ by the choice of ${\bf a}$ and ${\bf b}$.

Define the function $\nu$ on $D^n$ as follows.
Case (i):
If all ${\bf a},{\bf b},{{\bf a}\wedge_0 {\bf b}},{\bf a}\vee_0 {\bf b}$ and ${\bf a}\vee_1 {\bf b}$ are distinct, define
\begin{eqnarray}
\nu({\bf a}) &:=& \mu({\bf a}) - r,\nonumber\\
\nu({\bf b}) &:=& \mu({\bf b}) - r,\nonumber\\
\nu({{\bf a}\wedge_0 {\bf b}}) &:=& \mu({{\bf a}\wedge_0 {\bf b}}) + r,\nonumber\\
\nu({\bf a}\vee_0 {\bf b}) &:=& \mu({\bf a}\vee_0 {\bf b}) + r\cdot \al,\label{nu}\\
\nu({\bf a}\vee_1 {\bf b}) &:=& \mu({\bf a}\vee_1 {\bf b}) + r\cdot (1-\alpha),\nonumber\\
\mbox{and}\ \ \nu({\bf c}) &:=& \mu({\bf c})\ \ \mbox{otherwise}\nonumber.
\end{eqnarray}
If any of the five elements ${\bf a},{\bf b},{{\bf a}\wedge_0 {\bf b}},{\bf a}\vee_0 {\bf b}$ and ${\bf a}\vee_1 {\bf b}$ coincide, we have to add the corresponding adjustments as follows. Firstly note that, as ${\bf a}$ and ${\bf b}$ are incomparable, it is easy to see that at most one pair of two elements can coincide, and that there are only the following four possibilities for these two coinciding elements:
(ii) ${\bf a}\wedge_0 {\bf b} = {\bf a}\vee_0 {\bf b}$,
(iii) ${\bf a}\vee_0 {\bf b} = {\bf a}\vee_1 {\bf b}$,
(iv) ${\bf a}\vee_1 {\bf b} = {\bf a}$ and
(v) ${\bf a}\vee_1 {\bf b} =  {\bf b}$.

In case (ii), we define $\nu({{\bf a}\wedge_0 {\bf b}}) := \mu({{\bf a}\wedge_0 {\bf b}}) + r\cdot (1+\alpha)$ and all other function values as in \eqref{nu},
in case (iii), we define $\nu({{\bf a}\vee_0 {\bf b}}) := \mu({{\bf a}\vee_0 {\bf b}}) + r$ and all other function values as in \eqref{nu},
and in cases (iv) and (v), we define $\nu({{\bf a}\vee_1 {\bf b}}) := \mu({{\bf a}\vee_1 {\bf b}}) - r\cdot \al$ and all other function values as in \eqref{nu}.

The image of $\nu$ is in $[0,1]$ by the choice of $r$, and it is easy to check that in all five cases one has
$$\sum_{{\bf c} \in \{{\bf a},{\bf b},{{\bf a}\wedge_0 {\bf b}},{\bf a}\vee_0 {\bf b},{\bf a}\vee_1 {\bf b}\} }\nu({\bf c})
= \sum_{{\bf c} \in \{{\bf a},{\bf b},{{\bf a}\wedge_0 {\bf b}},{\bf a}\vee_0 {\bf b},{\bf a}\vee_1 {\bf b}\} }\mu({\bf c}).$$
This yields
$$\sum_{{\bf c} \in D^n}\nu({\bf c}) = \sum_{{\bf c} \in D^n}\mu({\bf c}) = 1,$$
so $\nu$ is a probability distribution. Furthermore, an easy calculation using Equation \eqref{eq:8} yields
$$\sum_{{\bf c} \in \{{\bf a},{\bf b},{{\bf a}\wedge_0 {\bf b}},{\bf a}\vee_0 {\bf b},{\bf a}\vee_1 {\bf b}\} }\nu({\bf c}){\bf c}
= \sum_{{\bf c} \in \{{\bf a},{\bf b},{{\bf a}\wedge_0 {\bf b}},{\bf a}\vee_0 {\bf b},{\bf a}\vee_1 {\bf b}\} }\mu({\bf c}){\bf c}$$
in all five cases, and so
$$\sum_{{\bf c} \in D^n}\nu({\bf c}){\bf c} = \sum_{{\bf c} \in D^n}\mu({\bf c}){\bf c}  = {\bf x},$$
so $\nu \in \PPPP({\bf x})$. The $\alpha$-bisubmodularity inequality \eqref{eq:1} yields
\begin{align*}
&\sum_{{\bf c} \in \{{\bf a},{\bf b},{{\bf a}\wedge_0 {\bf b}},{\bf a}\vee_0 {\bf b},{\bf a}\vee_1 {\bf b}\} }\mu({\bf c})f({\bf c}) -
 \sum_{{\bf c} \in \{{\bf a},{\bf b},{{\bf a}\wedge_0 {\bf b}},{\bf a}\vee_0 {\bf b},{\bf a}\vee_1 {\bf b}\} }\nu({\bf c})f({\bf c})   &=  \\
&r\left(
f({\bf a}) +
f({\bf b}) -
f({\bf a}\wedge_0 {\bf b}) -
\al f({\bf a}\vee_0 {\bf b}) -
(1-\alpha)f({\bf a}\vee_1 {\bf b})
\right)
&\stackrel{\eqref{eq:1}}{\ge} 0
\end{align*}
and so
$$\sum_{{\bf c} \in D^n}\nu({\bf c})f({\bf c}) \le \sum_{{\bf c} \in D^n}\mu({\bf c})f({\bf c}),$$
so $\nu \in \MMMM({\bf x})$.
Finally, we will show that
\begin{equation}\label{eq:4}
\sum_{{\bf c} \in D^n}\nu({\bf c})z^2({\bf c}) > \sum_{{\bf c} \in D^n}\mu({\bf c})z^2({\bf c}),
\end{equation}
which is a contradiction to the choice of $\mu$.
Let
\begin{align*}
A &:= \left|\menge{i \in [n]}{a_i = 0,\  b_i \neq 0}\right|,\\
B &:= \left|\menge{i \in [n]}{ b_i = 0,\  a_i \neq 0}\right|,\\
C &:= \left|\menge{i \in [n]}{ a_i =  b_i = 0}\right|\ \ \mbox{ and}\\
N &:= \left|\menge{i \in [n]}{0 \neq a_i \neq  b_i \neq 0}\right|.
\end{align*}
 The incomparability of ${\bf a}$ and ${\bf b}$ implies that we have either $N>0$ or, if $N=0$, we have both $A>0$ and $B>0$.
It is easy to check that
\begin{eqnarray*}
z({\bf a}\wedge_0 {\bf b})&=&A+B+C+N,\\
z({\bf a}\vee_0 {\bf b})&=&C+N,\\
z({\bf a}\vee_1 {\bf b})&=&C,\\
z({\bf a})&=&A+C,\\
z({\bf b})&=&B+C,
\end{eqnarray*}
and so
\begin{eqnarray*}
&& z({\bf a}\wedge_0 {\bf b})^2+\alpha\cdot z({\bf a}\vee_0 {\bf b})^2 + (1-\alpha)\cdot z({\bf a}\vee_1 {\bf b})^2 - z({\bf a})^2 - z({\bf b})^2\\
&=& (A+B+C+N)^2+\alpha(C+N)^2 + (1-\alpha)C^2 - (A+C)^2 - (B+C)^2\\
&=& 2(AB + AN + BN + CN) + N^2 + 2\al CN+\alpha N^2\\
&=& 2(AB + AN + BN + (1+\al) CN) + (1 + \al)N^2\ >\ 0,
\end{eqnarray*}
as $N>0$ or $AB>0$. As $r>0$ this implies
$$r(z({\bf a}\wedge_0 {\bf b})^2+\alpha\cdot z({\bf a}\vee_0 {\bf b})^2 + (1-\alpha)\cdot z({\bf a}\vee_1 {\bf b})^2 - z({\bf a})^2 - z({\bf b})^2)>0.$$
An easy calculation yields
$$\sum_{{\bf c} \in \{{\bf a},{\bf b},{{\bf a}\wedge_0 {\bf b}},{\bf a}\vee_0 {\bf b},{\bf a}\vee_1 {\bf b}\}}\nu({\bf c})z^2({\bf c}) > \sum_{{\bf c} \in \{{\bf a},{\bf b},{{\bf a}\wedge_0 {\bf b}},{\bf a}\vee_0 {\bf b},{\bf a}\vee_1 {\bf b}\}}\mu({\bf c})z^2({\bf c})$$
in all five cases for the definition of $\nu$.

From this, the contradicting inequality \eqref{eq:4} follows.
So $\mu$ is supported by a chain, and this implies $\mu = \la_{\bf x}$, which means that $f^L({\bf x}) = f^-({\bf x})$.

Thus $f^L = f^-$ holds and $f^L$ is convex.
\end{proof}


\vspace{1cm}
\begin{thebibliography}{10}

\bibitem{Bouchet}
Bouchet, A.: Greedy algorithm and symmetric matroids.
\newblock Mathematical Programming \textbf{38}, 147--159, 1987

\bibitem{BouchetCunningham}
Bouchet, A. and Cunningham, W.H.: Delta-matroids, jump systems and bisubmodular
  polyhedra.
\newblock SIAM J. Discrete Math. \textbf{8}, 17--32, 1995

\bibitem{ChandrasekaranKabadi}
Chandrasekaran, R. and Kabadi, S.N.: Pseudomatroids.
\newblock Discrete Math. \textbf{71}, 205--217, 1988

 \bibitem{Cohen06:soft}
 Cohen, D., Cooper, M., Jeavons, P., and Krokhin, A.: The complexity of soft constraint satisfaction.
 \newblock {\em Artificial Intelligence}, 170(11):983--1016, 2006.

\bibitem{Edmonds}
Edmonds, J.: Submodular functions, matroids, and certain polyhedra.
\newblock In: R.~Guy, H.~Hanani, N.~Sauer, J.~Sch{\"o}nheim (eds.)
  Combinatorial Structures and Their Applications, pp. 69--87. Gordon and
  Breach, 1970

\bibitem{Fujishige:book}
Fujishige, S.: Submodular Functions and Optimization.
\newblock Elsevier, 2005.

\bibitem{Fujishige13}
Fujishige, S., Tanigawa, S., and Yoshida, Y.: Generalized skew bisubmodularity:
A characterization and a min-max theorem.
\newblock Technical Report RIMS-1781, Research Institute for Mathematical Sciences,
Kyoto University, 2013.

\bibitem{Groetschel}
Gr\"{o}tschel, M., Lov\'asz, L., and Schrijver, A.: The ellipsoid method and its consequences in combinatorial optimization.
\newblock Combinatorica, 1, 169-197, 1981.


\bibitem{ours:soda}
Huber, A., Krokhin, A., and Powell, R.: Skew {B}isubmodularity and {V}alued {C}{S}{P}s.
\newblock In {\em Proceedings of {SODA'13}}, pages 1296-1305, 2013.

\bibitem{KabadiChandrasekaran}
Kabadi, S.N. and Chandrasekaran, R.: On totally dual integral systems.
\newblock Discrete Appl. Math. \textbf{26}, 87--104, 1990

\bibitem{KrokhinLarose}
Krokhin, A. and Larose, B.: Maximizing supermodular functions on product lattices,
  with application to maximum constraint satisfaction.
\newblock SIAM Journal on Discrete Mathematics, 22(1), 312-328, 2008.

\bibitem{Kuivinen}
Kuivinen, F.: On the complexity of submodular function minimisation on
               diamonds.
\newblock Discrete Optimization, 8(3), 459-477, 2011.
 

\bibitem{Lovasz}
Lov\'asz, L.: Submodular functions and convexity.
\newblock In A.~Bachem, M.~Gr\"otschel, and B.~Korte, editors, {\em
  Mathematical Programming: The State of the Art}, pages 235--257. Springer,
  1983.

\bibitem{McCormick}
McCormick, S.T.: Submodular function minimization.
\newblock In: K.~Aardal, G.~Nemhauser, and R.~Weismantel, editors, {\it Handbook on Discrete Optimization},
pages 321--391. Elsevier, 2006.


\bibitem{Thapper12:power}
Thapper, J. and \v{Z}ivn\'y, S.: The power of linear programming for valued {CSP}s.
\newblock In {\em Proceedings of {FOCS'12}}, pages 669-678, 2012.

\bibitem{Qi}
Qi, L.: Directed submodularity, ditroids and directed submodular flows.
\newblock {\em Mathematical Programming}, 42(1--3):579--599, 1988.

\bibitem{Schrijver}
Schrijver, A.: Combinatorial Optimization: Polyhedra and Efficiency.
\newblock Springer, 2004

\end{thebibliography}
\end{document}